\documentclass[lettersize,journal]{IEEEtran}
\usepackage{amsmath,amsfonts,amssymb} 
\usepackage{amsthm}
\usepackage{algorithmic}
\usepackage{algorithm}
\usepackage{array}
\usepackage{subfigure}
\usepackage{textcomp}
\usepackage{stfloats}
\usepackage{url}
\usepackage{verbatim}
\usepackage{graphicx}
\usepackage{cite}
\usepackage[hidelinks]{hyperref}

\newtheorem{assumption}{Assumption}
\newtheorem{remark}{Remark}
\newtheorem{theorem}{Theorem}
\newtheorem{proposition}{Proposition}
\newtheorem{problem}{Problem}
\usepackage{xcolor}
\usepackage[flushleft]{threeparttable}
\usepackage{bm}

\begin{document}
	
\title{Safety-based Speed Control of a Wheelchair using Robust Adaptive Model Predictive Control} 
	
\author{Meng Yuan, Ye Wang, Lei Li, Tianyou Chai, Wei Tech Ang}

\maketitle

\begin{abstract}
    Electric-powered wheelchair plays an important role in providing accessibility for people with mobility impairment. Ensuring the safety of wheelchair operation in different application scenarios and for diverse users is crucial when the designing controller for tracking tasks. In this work, we propose a safety-based speed tracking control algorithm for wheelchair systems with external disturbances and uncertain parameters at the dynamic level. The set-membership approach is applied to estimate the sets of uncertain parameters online and a designed model predictive control scheme with online model and control parameter adaptation is presented to guarantee safety-related constraints during the tracking process. The proposed controller can drive the wheelchair speed to a desired reference within safety constraints. For the inadmissible reference that violates the constraints, the proposed controller can steer the system to the neighbourhood of the closest admissible reference. The effectiveness of the proposed control scheme is validated based on the high-fidelity speed tracking results of two tasks that involve feasible and infeasible references. 
\end{abstract}


\section{Introduction}

\IEEEPARstart{W}{heelchairs} are essential devices in providing mobility for elderly and physically impaired people including patients with a spinal cord injury and stroke patients \cite{Haisma2006,Jang2016,Candiotti2019}. Among different types of wheelchairs, the electric-powered wheelchair has seen increasing popularity due to its convenience compared to manual wheelchairs. According to a survey study in \cite{Fehr2000}, around $80\%$ of electric wheelchair users rely on the joystick to maneuver the wheelchair. However, the unmodified reference generated by the joystick may result in unexpected high speed and acceleration with safety issues, and advanced speed control is generally required.

As a system with nonholonomic constraints and a clear model of the mechatronics system, the tracking control of wheelchairs has been widely studied by many researchers \cite{lian2015near,li2017adaptive,li2021trajectory}. Depending on the type of model used and the level of control implemented, the speed control of wheelchairs can be divided into two categories. The first category is related to the kinematic model of the system where the motion of interest fails to satisfy Brockett’s necessary conditions \cite{Zhang2014}. Given the desired Cartesian position and orientation of the wheelchair, the control objective is to design the linear and angular velocities such that the wheelchair tends to the given position with the required orientation. In \cite{Zhang2014}, a finite-time tracking of the wheelchair based on cascaded control architecture and sliding model control was presented. In \cite{DeLaCruz2011}, an adaptive tracking controller was designed for a wheelchair with the input-to-state stability guarantee.

Although some existing algorithms, e.g. the timed elastic band-based method, can minimize the execution time of trajectory while considering the kinodynamic constraints, the tracking control at the kinematic level is under the assumption of perfect velocity tracking at the dynamics level, which may not be realistic in practical applications \cite{Fu2013}. The ignored disturbances and system uncertainties can cause severe constraint violation even if the velocity and acceleration tolerances are considered at the kinematic level when conducting both planning and control.

The second class of tracking control for nonholonomic systems involves the dynamics of actuators. The objective of this type of control is to design the current or voltage at the dynamics level to ensure the convergence of the system state to desired Cartesian position or wheel velocities. In \cite{Fu2013}, the trajectory tracking of nonholonomic system is achieved by force control using cascaded and back-stepping techniques. Later, the authors extend their controller for motion tracking of a mobile robot at the voltage level with a simplification based on a linear relation between velocity and voltage \cite{Fu2020}. Based on Lyapunov and back-stepping methods, \cite{Do2004} proposed an adaptive control law for stabilizing and tracking of nonholonomic robots with unknown system parameters. 

However, none of these methods discussed above guarantee the safety-related constraints of wheelchairs in the presence of disturbances. This is an important and common problem that can be found in scenarios when the wheelchair is required to operate within some given speed and acceleration constraints while the daily-life tasks introduce external disturbances such as driving the wheelchair on inclined ramps. Moreover, the advent of rentable or sharing wheelchairs at big-city hospitals makes it desirable to design a controller that can ensure safety-related constraints at the
dynamics level while considering the variation of system parameters due to the change of users \cite{Zhecheng}.

Model predictive control (MPC), as an optimal control strategy capable of explicitly handling the operation constraints, has been widely used in applications with strict demands on state and input constraints \cite{Mayne2000,Mayne2014}. To ensure robust constraint satisfaction for systems with parameter variations, in \cite{Fleming2015}, a robust tube MPC is designed for a linear parameter varying (LPV) system and the online computation load is reduced by constructing a terminal set involving the norm bounds of tube parameters. In \cite{Zhang2020}, a recursive least square-based adaptive MPC is designed for constrained systems with unknown model parameters. The tube-based adaptive MPC provides less conservative performance compared with the robust tube MPC. For systems with both parameter variations and external disturbances, robust adaptive MPC with the set-membership approach is adopted to provide robust constraint satisfaction with online parameter update \cite{Tanaskovic2019b,Lorenzen2019,Lu2021,Bujarbaruah2022}. However, most of the existing works of robust adaptive MPC were designed for the regulation problem and few works can be found for the tracking problem of systems with constraints in the presence of external disturbances and uncertain parameters. 

Inspired by the timely needs of operating wheelchairs safely in different application scenarios and for diverse users, the main contribution of this work is to present a safety-based constrained tracking control algorithm for wheelchair systems with external disturbances and uncertain parameters at the dynamics level. Under the assumption of unknown-but-bounded disturbances, the set-membership approach is adopted to provide an updated set of uncertain parameters which is used in the subsequent MPC design. For any infeasible reference that violates the safety constraints, an optimization-based method is utilized to compute the closest admissible reference for tracking. The state and input constraints are robustly satisfied in the proposed MPC framework with reference-dependent and tube-based constraints in vertex representation of states. The recursive feasibility and input-to-state stability of the wheelchair system with the proposed MPC controller are guaranteed. The effectiveness of the proposed safety-based control algorithm is validated by two speed tracking tasks on the high-fidelity model of a practical wheelchair.

The remainder of this paper is organized as follows: In Section~\ref{sec:sys_prob}, the dynamics of the wheelchair system at actuator level is described and the system in the LPV form with unknown parameters is discussed. The proposed robust adaptive tracking controller with details in parameter estimation, feasible reference generation, robust constraint satisfaction and tracking MPC formulation are presented in Section \ref{sec:controller}. The tracking results of the proposed safety-based control algorithm are demonstrated in Section~\ref{sec:results}. Section~\ref{sec:conclusion} concludes this work. 

\emph{Notations}: When defining the variable, we follow the rule that capitalized letters are for matrices and small letters are for vectors or scalars. $\mathbb{R}$ and $\mathbb{Z}$ are the sets of real and integer numbers. Given two integers $a$, $b\in \mathbb{Z}$, $\mathbb{Z}_{a+} \triangleq \{ i\in \mathbb{Z} \,|\, i \geq a\}$ and $\mathbb{Z}_{[a,b]} \triangleq \{i \in \mathbb{Z} \,|\, a\leq i \leq b\}$. The $i$-th row of matrix $X$ and the $i$-th element of vector $x$ are represented by $X^{[i]}$ and $x^{[i]}$, respectively. The $i$-th vertex of $x \in \mathcal{X}$ is denoted by $x^{(i)}$. A non-negative matrix is denoted by $X \geq 0$. The positive definite and semi-definite matrices are represented as $X \succ 0 $ and $X\succeq 0 $, respectively. For a matrix $X \in \mathbb{R}^{n \times n}$, the smallest eigenvalue is denoted by $\underline{\lambda}(X)$. The identity matrix of dimension $n$ is denoted by $I_{n}$ and an $m$-dimension vector with all elements as 1 is denoted by $\mathbf{1}_{m}$. The $m \times n$ matrix with all elements as zero is denoted by $\mathbf{0}_{m,n}$. A diagonal matrix with main diagonal elements $a_1,\ldots,a_n$ is denoted by $\mathrm{diag}(a_1, \ldots, a_n)$. The following sets are defined: $ \mathbb{S}^{n} = \{ X \in \mathbb{R}^{n \times n}: X = X^\top \} $, $ \mathbb{S}^{n}_{\succ 0} = \{ X \in \mathbb{S}^{n}: X \succ 0 \} $ and $ \mathbb{S}^{n}_{\succeq 0} = \{ X \in \mathbb{S}^{n}: X \succeq 0 \}$. A convex polyhedral set of $x$ is defined as $ \mathcal{P}_{x}(F_{x},b_{x})= \{ x \,| \, F_{x}x \leq b_{x}\} $. With $P \in \mathbb{S}^{n}_{\succ 0}$, an ellipsoidal set of $x$ is defined as $\mathcal{E}(P,1) = \{ x \,|\, x^{\top} P x \leq 1 \}$. For a vector $x\in \mathbb{R}^{n}$ with a matrix $Q$, $\| x \|$ denotes the 2-norm and $\| x \|_{Q}$ stands for $\sqrt{x^{\top}Qx}$. The vector $x_{i|k}$ represents the predicted value of $x$ at a sampling time instant $k+i$ based on measurement at $k$. The vector $x(k)$ stands for the measured value of $x$ at a sampling time instant $k$. A continuous function $f: [0,a) \rightarrow [0,\infty)$ belongs to class $\mathcal{K}_{\infty}$ if $a=\infty$ and $f(r) \rightarrow \infty$ as $r\rightarrow \infty$. 

\section{System description and problem formulation}\label{sec:sys_prob}
	
\subsection{Wheelchair Dynamics}	

With implicit variable change for rotary to translational movement, the electrical and mechanical subsystems of electric powered wheelchair are given as follows \cite{Do2004}:
\begin{subequations}
    \begin{align}
       & L\dot{i}_c+ Ri_{c} +K_{e}v = u, \label{eq:ele_wheel}\\
       & M\dot{v} + Dv + w_{f} = K_{t}i_c, \label{eq:mecha_wheel}
    \end{align}
\end{subequations}
where $v = [v_{1},v_{2}]^{\top}$ is the vector of the linear velocities, $u = [u_{1},u_{2}]^{\top}$ is the vector of the motor voltages, $i_c = [i_{c_1}, i_{c_2}]^{\top}$ is the vector of the currents and $w_{f} = [w_{1}, w_{2}]^{\top}$ is the vector of the lumped disturbance torques on the right and left wheels, respectively. $M$ and $D$ are the equivalent mass and damping coefficient matrices, which can be expressed as

\[
M=\left[\begin{array}{cc}
	m_{11} & m_{12}\\
	m_{21} & m_{22}
\end{array}\right], \, D= \mathrm{diag}(d_{1},d_{2}),
\]
where $m_{11}$, $m_{12}$, $m_{21}$ and $m_{22}$ are scalars. The non-zero matrix $M$ couples the dynamics of left and right wheels, and $d_{1}$ and $d_{2}$ are the damping coefficients for the right and left wheels, respectively. Furthermore, $L = \mathrm{diag}(l_1, l_2)$ and $R = \mathrm{diag}(r_1,r_2)$ are the inductance and resistance matrices of system, respectively. $K_{e}= \mathrm{diag}(k_{e_1}, k_{e_2})$ and $K_t = \mathrm{diag}(k_{t_1}, k_{t_2})$ are the back electromotive force constant and torque constant matrices, respectively.

For system without current feedback $i_c$, the dynamics in \eqref{eq:ele_wheel} and \eqref{eq:mecha_wheel} can be reformulated as
\begin{equation}
	\label{eq:dyn_sys}
	\ddot{v} + (M^{-1}D+L^{-1}R)\dot{v} + \Gamma v = M^{-1} L^{-1} K_{t}u + w,
\end{equation}
where $w = -M^{-1}\dot{w}_{f} - M^{-1}L^{-1}Rw_{f}$ and $\Gamma = M^{-1}L^{-1}RD+M^{-1}L^{-1}K_{t}K_{e}$ are the lumped terms.
	
Let $x \triangleq [v^{\top},\dot{v}^{\top}]^{\top}$ be the state vector. By using the Euler forward approximation, the discrete-time state-space model with a sampling time interval $T_{s}$ can be formulated as

\begin{align}
	\label{eq:sys_discrete}
	x(k+1) & =A x(k)+B u(k)+E w(k) \nonumber\\
	& =\left[\begin{array}{cc}
		I_{2} & T_{s}I_{2}\\
		-T_{s}\Gamma & I_{2}-T_{s}(M^{-1}D+L^{-1}R)
	\end{array}\right]x(k) \nonumber\\
	& \;\; +\left[\begin{array}{c}
		\mathbf{0}_{2,2}\\
		T_{s}M^{-1}K_{t}L^{-1}
	\end{array}\right]u(k)+\left[\begin{array}{c}
		\mathbf{0}_{2,2}\\
		T_{s} I_{2}
	\end{array}\right] w(k).
\end{align}

All the states $x(k)$ and inputs $u(k)$ are required to satisfy the following constraints to guarantee the safety of the wheelchair:
	\begin{equation}
	\label{eq:const_state_input}
		Gx(k) + Hu(k) \leq b,
	\end{equation}
with given matrices $G \in \mathbb{R}^{n_{c}\times 4}$, $H \in \mathbb{R}^{n_{c}\times 2}$ and $b \in \mathbb{R}^{n_{c}}$.

\begin{remark}
    The safety constraints in \eqref{eq:const_state_input} are defined in a general form, where $b$ is not necessarily a vector with all elements of 1. This form offers more degrees of freedom when configuring state and input sets, and indicates that the coordinate origin may not be inside both state and input constraint sets. 
\end{remark}

\begin{assumption}\label{assumption:bounded disturbances}
    For the wheelchair dynamics \eqref{eq:sys_discrete}, the disturbance vector $w(k)$ is unknown but can be bounded by a convex polyhedral set as
    \begin{equation}
    	w(k) \in \mathcal{P}_w (\bar{F}_{w},\bar{b}_{w}), \forall k \in \mathbb{Z}_{0+},
    \end{equation}
    with $\bar{F}_{w} \in \mathbb{R}^{n_{w}\times 2}$ and $\bar{b}_{w} \in \mathbb{R}^{n_{w}}$.
    
\end{assumption}

\subsection{Wheelchair LPV Model}

Due to the fact that the wheelchair would be used by diverse users, there exist time-varying uncertain parameters in the wheelchair dynamics \eqref{eq:sys_discrete}. Therefore, the model \eqref{eq:sys_discrete} is reformulated in an LPV form as follows:
\begin{equation}
\label{eq:sys_dis_vary}
    x(k+1) = A(\theta) x(k)+B(\theta) u(k)+E w(k),
\end{equation}
where $A(\theta)$ and $B(\theta)$ are parameter-varying affine matrices with unknown parameter $\theta \in \mathbb{R}^{q}$ as
\begin{equation}
	\label{eq:sys_para_variation}
	A(\theta) = A_{0} + \sum_{j=1}^{q} A_{j} \theta^{[j]}, \; B(\theta) = B_{0} + \sum_{j=1}^{q} B_{j} \theta^{[j]},
\end{equation}
where $q$ is the number of unknown parameters and the matrices $A_{j}$, $B_{j}$, $j\in \mathbb{Z}_{[1,q]}$ are known.

\begin{assumption}
    For the wheelchair application, the parameter $\theta$ may involve the varying of mass and damping coefficient due to the change of operators. The true value of uncertain parameter $\theta^{*}$ is unknown, piecewise constant but can be bounded by a polyhedral set $\theta^* \in \mathcal{P}_{\theta}(\bar{F}_{\theta},\bar{b}_{\theta})$, where $\bar{F}_{\theta}\in \mathbb{R}^{n_{\theta} \times q}$ and $\bar{b}_{\theta} \in \mathbb{R}^{n_{\theta}}$.
\end{assumption}

With given piecewise constant speed references $v_{d}$ representing the desired speeds of right and left wheels, this work is dedicated to solve the following problem:
\begin{problem}
	For the nonholomonic wheelchair in the LPV form of \eqref{eq:sys_dis_vary} with \eqref{eq:sys_para_variation}, design a tracking MPC law $u = \kappa(x,v_{d})$ to steer the wheelchair speed to the desired piecewise constant reference $v_{d}$ while the state and input constraints \eqref{eq:const_state_input} are always satisfied in the presence of parameter uncertainties and external disturbances.
\end{problem}

\section{Controller architecture}\label{sec:controller}

In this section, we propose a robust adaptive MPC for achieving the safety-based speed tracking control of the wheelchair with unknown parameters and external disturbances. The set-membership approach is used to provide an updated polytopic set of unknown parameters and the estimated value is updated based on the Euclidean projection. An admissible state and input are introduced when infeasible reference is given in speed tracking control. It is followed by the construction of polytopic tubes for ensuring the state and input constraints. Finally, an MPC tracking control scheme is presented with terminal cost and terminal constraint. The closed-loop properties of the proposed control algorithm are provided.  

\subsection{Set-membership Parameter Estimation}

To use the wheelchair LPV model in the MPC design, an estimation method for unknown uncertain parameters is required. As suggested by \cite{Lorenzen2019,Chisci1998}, the set-membership parameter estimation can provide a quantification of associate uncertainties and contribute to the MPC design, which is adopted in this work. 

For notation simplicity, we denote the uncertain regressor vector $\Phi \in \mathbb{R}^{4 \times q}$ and the certain regressor vector $\phi \in \mathbb{R}^{4\times 1}$ as
\begin{subequations}
    \begin{align}
        \Phi(x,u) &= \left[ A_{1} x + B_{1} u, \cdots, A_{q}x+B_{q}u \right],\\
        \phi(x,u) &= A_{0} x + B_{0} u.
    \end{align}
\end{subequations}

Let $\bar{F}_{\theta} \in \mathbb{R}^{n_{\theta} \times q}$ be a pre-defined time-invariant matrix, the estimated parameter at time instant $k$, i.e., $\hat{\theta}(k)$, is bounded in a polytopic set:
\begin{equation}\label{eq:theta est set}
    \hat{\theta}(k) \in \Theta(k) = \mathcal{P}_{\theta}(\bar{F}_{\theta},b_{\theta}(k)).
\end{equation}

Then, the uncertain parameter bounding set can be iteratively computed by 
\begin{equation}\label{eq:theta set update}
    \Theta(k) = \Theta(k-1) \cap \Delta_\Theta (k), \, k \in \mathbb{Z}_{1+},
\end{equation}
where $ \Delta_\Theta(k) =  \{\theta \,|\, x(k)-A(\theta)x(k-1)-{B(\theta) u(k-1)} \in \mathcal{P}(\bar{F}_{w},\bar{b}_{w})\} = \{\theta \,|\, -\bar{F}_{w}\Phi(k-1)\theta \leq \bar{b}_{w}+\bar{F}_{w}\left( \phi(k-1)-x(k) \right)\}$ and $\Theta(0) = \mathcal{P}_{\theta}(\bar{F}_{\theta}, \bar{b}_{\theta})$.

Since the constant matrix $\bar{F}_{\theta}$ defines fixed directions of half spaces, the update of $\Theta(k)$ is equivalent to updating $b_{\theta}(k)$, which can be obtained by solving the following optimization problem \cite[Proof of Propoisition 3.31]{Blanchini2015}:

\begin{subequations}
    \begin{align}
        b_{\theta}^{[l]}(k) & =\underset{\varLambda^{[l]}}{\min} \,\varLambda^{[l]}\left[\begin{array}{c}
        b_{\theta}(k-1) \\
        \bar{b}_{w}+\bar{F}_{w}\left(\phi(k-1)-x(k)\right)
        \end{array}\right],\\
        \text{subject to} & \;\varLambda^{[l]}\left[\begin{array}{c}
        \bar{F}_{\theta},\\
        -\bar{F}_{w}\Phi(k-1)
        \end{array}\right]=\bar{F}_{\theta}^{[l]}, \\
         & \; \varLambda^{[l]}\geq 0,
\end{align}
\end{subequations}
for $l \in \mathbb{Z}_{[1,n_{\theta}]}$ , where $\varLambda \in \mathbb{R}^{n_{\varLambda}\times n_{\varLambda}}$, $n_{\varLambda}=n_{\theta}+n_{w}$. At time instant $k \in \mathbb{Z}_{1+}$, the estimated value of uncertain parameters is computed by a projection of $\hat{\theta}(k-1)$ to the set $\Theta(k)$:
\begin{equation}
\label{eq:proj_theta}
\hat{\theta}(k) = \underset{\theta\in \Theta(k)}{\arg\min} \, \| \theta - \hat{\theta}(k-1) \|_{2} ,
\end{equation}
with a given initial $\hat{\theta}(0)$.

\begin{remark}
    Other parameter estimation methods such as Kalman filters or least mean squares filters can be used to generate an intermediate value of uncertain parameter before projecting $\hat{\theta}(k-1)$ to $\Theta(k)$ \cite{Lorenzen2019}. Here, the direct projection \eqref{eq:proj_theta} is essential in the closed-loop property analysis of the proposed control algorithm.
\end{remark}
	
\subsection{Controller Design}

\subsubsection{Generation of Admissible Reference}

Under certain circumstances, such as operating the wheelchair manually with joystick input, the given reference $v_{d}$ may not be admissible since the state and input constraints for safety purpose cannot be satisfied. The proposed controller is required to steer the wheelchair to the closest admissible reference characterized by introducing a pair of steady state and input $(x_{s},u_{s})$ as:
\begin{subequations}\label{prob:xs_us}
    \begin{align}
        (x_{s}(k),u_{s}(k)) & = \underset{x_{s},u_{s}}{\arg\min}\left\Vert v_{d}(k)-y_{s}\right\Vert _{Q_{s}}^{2},\\
        \text{subject to} & \;\; x_{s}=A(\hat{\theta}(k))x_{s}+B(\hat{\theta}(k))u_{s},\\
        & \;\; y_{s} = C_{s} x_{s}, \\
        & \;\; Gx_{s}+Hu_{s}\leq b,
\end{align}
\end{subequations}
where $C_{s} = [I_{2},\mathbf{0}_{2,2}]$, and the weighting matrix $Q_{s}$ can be chosen as $I_{2}$ to equally penalise the speed difference.

\subsubsection{Synthesis of Local and Terminal Control Gains}

To achieve speed tracking using MPC with $N$-step prediction horizon, an error state $e_{i|k} = x_{i|k} - x_{s}(k)$ and virtual control input $\bar{u}_{i|k} = u_{i|k} - u_{s}(k)$ are introduced, and the error coordinate follows: 
\begin{equation}
    e_{i+1|k} = A(\hat{\theta}(k))e_{i|k} + B(\hat{\theta}(k)) \bar{u}_{i|k}, \, i \in \mathbb{Z}_{[0,N]}.
\end{equation}

The control sequence at time instant $k$ is designed below to deal with the open-loop and closed-loop mismatch caused by uncertain parameters and external disturbances:
\begin{align}\label{eq:control sequence}
    \bar{u}_{i|k} =  \begin{cases}
        K e_{i|k} + \mu_{i|k}, & i \in \mathbb{Z}_{[0,N-1]},\\
        K_f e_{i|k}, & i = N,
    \end{cases}
\end{align}
where $\mu_{i|k}$, $i \in \mathbb{Z}_{[0,N-1]}$ are control parameters in the MPC; $K \in \mathbb{R}^{2 \times 4}$ and $K_f \in \mathbb{R}^{2 \times 4}$ are local and terminal control gains, respectively. The following assumptions are useful for finding the control gains $K$ and $K_{f}$.

\begin{assumption}(Local Control Gain $K$)\label{assumption:K}
For matrix $\bar{F}_{e} \in \mathbb{R}^{n_{e}\times 4}$ and all $e$ that satisfies $\bar{F}_{e}e\leq \mathbf{1}_{n_{c}}$, there exists a local control gain $K$ such that $\bar{F}_{e}(A(\theta)+B(\theta)K)e\leq \varepsilon \mathbf{1}_{n_{c}}$ with a $\varepsilon \in [0,1)$ for all $\theta \in \mathcal{P}_{\theta}(\bar{F}_{\theta},\bar{b}_{\theta})$.
\end{assumption}

\begin{assumption} (Terminal Control Gain $K_{f}$)\label{assumption:Kf}
    For all $\theta \in \mathcal{P}_{\theta}(\bar{F}_{\theta}, \bar{b}_{\theta})$, given matrices $Q \in \mathbb{S}^{4}_{\succeq 0}$ and $R \in \mathbb{S}^{2}_{\succ 0}$, there exists a positive definite matrix $P \in \mathbb{S}^{4}_{\succ 0}$ and a terminal state feedback gain $K_f$ such that $P - A_c(\theta) ^{\top} P A_c(\theta) - Q - K_f^\top R K_f \succeq 0$ and $ e \in \mathcal{E}(A_c(\theta)^{\top} P A_c(\theta),\varepsilon_f) $ with a $\varepsilon_f \in [0,1)$ for all $ e \in \mathcal{E}(P,1)$, where $A_c(\theta) = A(\theta) + B(\theta)K_{f} $.
\end{assumption}

\begin{remark}
   
    The matrix $\bar{F}_{e}$ can be chosen such that the polytopic set $\mathcal{P}_{e}(\bar{F}_{e},\mathbf{1}_{n_{c}})$ approximating a robust control invariant set of state formed by $\mathcal{E}(\tilde{P},1) = \{ e \,|\, e^{\top} \tilde{P} e \leq 1 \}$, where $\tilde{P}$ can be chosen by satisfying the condition $\tilde{P} - (A(\theta) + B(\theta) K)^{\top}\tilde{P}(A(\theta) + B(\theta)K) \succ 0$. Then, the value of $K$ can be computed by minimizing $\varepsilon$ while subjecting to set inclusion with $\mathcal{P}_{e}(\bar{F}_{e},\mathbf{1}_{n_{e}}) \subseteq \mathcal{P}_{e}(\bar{F}_{e}(A(\theta)+B(\theta)K),\varepsilon \mathbf{1}_{n_{e}})$ \cite{Blanchini2015}.
\end{remark}

\begin{remark}\label{remark:Kf=K}
    The terminal control gain $K_f$ can be chosen by satisfying the condition described in Assumption \ref{assumption:Kf}. Alternatively, since the local control gain $K$ is chosen from a contractive polyhedral set, there may also exist a contractive ellipsoidal set as discussed in Assumption \ref{assumption:Kf} with this gain. So the terminal control gain $K_f$ can be chosen to be the same as $K$.
\end{remark}

\subsubsection{Tube-based Constraint Satisfaction}

To ensure the robust constraints satisfaction based on the predicted error state and input, we start by constructing a sequence of time-varying sets along the MPC prediction horizon for the tracking error state:
\begin{align}\label{eq:polytopic_tube}
    \bar{F}_e e_{i|k} \leq \alpha_{i|k}, \; i \in \mathbb{Z}_{[0,N]},
\end{align}
where $\alpha_{i|k} \in \mathbb{R}^{n_{e}\times 1}$ are another decision variables in the MPC. 

With the given $\bar{F}_{e}$, the set of error state in \eqref{eq:polytopic_tube} can be represented using the vertices of each set as:
\begin{align}\label{eq:eAndS}
    e_{i|k}^{(j)} = S_j \alpha_{i|k}, \; j \in \mathbb{Z}_{[1,p]},
\end{align}
where $p$ is the total number of vertices of the tracking error set. For each vertex with index $j$, the following equality holds
\begin{equation}
\label{eq:active_tube}
    \bar{F}_{e}^{[l]}e_{i|k}^{(j)} = \alpha_{i|k}^{[l]}, j\in \mathbb{Z}_{[1,p]}, l \in \mathcal{R}_{j},
\end{equation}
where $\mathcal{R}_{j}\in \mathbb{R}^{4} \subset \mathbb{Z}_{[1,n_c]}$ is the index set for the vertex with index $j$ that describes the active rows of inequality in \eqref{eq:polytopic_tube}. 

From \eqref{eq:eAndS} and \eqref{eq:active_tube}, we can see that the $S_j$ is independent to $\alpha_{i|k}$ and can be computed based on:
\begin{equation}
    \bar{F}_{e}^{[l]}S_{j} = I_{n_c}^{[l]}, j \in \mathbb{Z}_{[1,p]}, l \in \mathcal{R}_{j}.
\end{equation}

For all $e$ satisfying \eqref{eq:polytopic_tube}, $\theta\in \Theta(k)$ and $w\in \mathcal{P}_{w}(\bar{F}_{w},\bar{b}_{w})$, at the next prediction time step $i+1$, the following condition should also hold, for $i\in \mathbb{Z}_{[0,N-1]}$,
\begin{equation}
\label{eq:tube_i_1}
    \bar{F}_e \left( A(\hat{\theta}(k))e_{i|k} + B(\hat{\theta}(k)) \bar{u}_{i|k} \right) + \bar{w} \leq \alpha_{i+1|k},
\end{equation}
where the vector $\bar{w}$ has its elements $w^{[l]} = \underset{w\in \mathcal{P}_{w}}{\max}(\bar{F}_{e}w)^{[l]}$ for $l\in \mathbb{Z}_{[1,n_{w}]}$. 

\begin{proposition}[Polytopic Set Inclusion \cite{Blanchini2015}]\label{pro:set inclusion}
    For two polytopic set $\mathcal{P}_1 (F_1,b_1)$ and $\mathcal{P}_2 (F_2,b_2) $, the inclusion $\mathcal{P}_1 (F_1,b_1) \subseteq \mathcal{P}_2 (F_2,b_2)$ holds if and only if there exists a non-negative matrix $\Omega \geq 0$ such that
    \begin{subequations}
        \begin{align}
            & \Omega F_1  = F_2,\\
            & \Omega b_1 \leq b_2.
        \end{align}
    \end{subequations}
\end{proposition}

Based on Proposition \ref{pro:set inclusion}, for all $\hat{\theta}(k) \in \Theta(k)$ from \eqref{eq:theta est set}, the inequality \eqref{eq:tube_i_1} is reformulated as
\begin{align}\label{eq:MPC theta condition}
    \bar{F}_e \Phi & \left( S_j \alpha_{i|k}, K S_j \alpha_{i|k}  + \mu_{i|k}  \right) \hat{\theta}(k) \nonumber\\
     &+ \bar{F}_e \phi \left( S_j \alpha_{i|k}, K S_j \alpha_{i|k} + \mu_{i|k}  \right) + \bar{w} \leq \alpha_{i+1|k}.
\end{align}

By using Proposition \ref{pro:set inclusion}, the condition \eqref{eq:MPC theta condition} for all $\hat{\theta}(k)\in \Theta(k)$ can be reformulated with vertex representation as follows:
\begin{subequations}\label{eq:vertex tube constraints}
    \begin{align}
        & \Omega_{j,i|k} \bar{F}_{\theta} = \bar{F}_e \Phi \left( S_j \alpha_{i|k}, K S_j \alpha_{i|k} + \mu_{i|k}  \right),\label{eq:vertex tube constraints-1}\\
        & \Omega_{j,i|k} b_{\theta}(k) + \bar{F}_e \phi \left( S_j \alpha_{i|k}, K S_j \alpha_{i|k} + \mu_{i|k}  \right)  \nonumber\\
        & \qquad - \alpha_{i+1|k} \leq -\bar{w},\label{eq:vertex tube constraints-2}\\
        & \Omega_{j,i|k} \geq 0, \, j \in \mathbb{Z}_{[1,p]}, \,i \in \mathbb{Z}_{[0,N-1]},\label{eq:vertex tube constraints-3}
    \end{align}
\end{subequations}
where $\Omega_{j,i|k} \in \mathbb{R}^{n_e \times q}$ for each $i$. 

To satisfy the original state and input constraints described in \eqref{eq:const_state_input}, the following constraint should be enforced in the MPC formulation:
\begin{align*}
    G e_{i|k} + H \bar{u}_{i|k} \leq \tilde{b}(k), \; i\in \mathbb{Z}_{[0,N]},\, \forall k \in \mathbb{Z}_{0+},
\end{align*}
where $\tilde{b}(k)=b - G x_{s}(k) - H u_{s}(k)$. Then, the above condition is equivalent to
\begin{align}\label{eq:vertex state input constraint}
     (G + H K) S_j \alpha_{i|k} + H \mu_{i|k} \leq \tilde{b}(k).
\end{align}

Following the similar procedure described above with \eqref{eq:control sequence}, the terminal conditions can also be formulated as follows:
\begin{subequations}
    \begin{align}
        & \Omega_{j,N|k} \bar{F}_{\theta} = \bar{F}_e \Phi \left( S_j \alpha_{N|k}, K_f S_j \alpha_{N|k} \right),\label{eq:vertex terminal constraint-1}\\
        & \Omega_{j,N|k} b_{\theta}(k) + \bar{F}_e \phi \left( S_j \alpha_{N|k}, K_f S_j \alpha_{N|k} \right) \nonumber\\ 
        & \qquad -\alpha_{N|k}  \leq -\bar{w},\label{eq:vertex terminal constraint-2}\\
        & \Omega_{j,N|k} \geq 0, \, j\in \mathbb{Z}_{[1,p]},\label{eq:vertex terminal constraint-3} \\
        & (G + H K) S_j \alpha_{N|k} +  H \sigma(k) \leq \tilde{b}(k). \label{eq:vertex terminal constraint-4}
    \end{align}
\end{subequations}

\subsubsection{Optimization Formulation}

For given speed reference $v_{d}(k)$, admissible steady states and inputs $(x_s(k),u_s(k))$ obtained by \eqref{prob:xs_us}, the safety-based speed tracking MPC optimization problem can be formulated as follows:
\begin{subequations}\label{problem:RAMPC}
	\begin{align}
	\underset{\mu_{i | k},\alpha_{i|k},\Omega_{j,i|k}}{\min}& \;
	\sum_{i=0}^{N-1} \ell \left( e_{i|k},\bar{u}_{i|k} \right) + \| e_{N|k} \|_{P(\hat{\theta}(k))}^2, \label{eq:RAMPC-cost function}\\
	\text{subject to} & \quad e_{0|k} = x(k)-x_s(k),\\
	& \quad e_{i+1|k} =A(\hat{\theta}(k)) e_{i|k} + B(\hat{\theta}(k)) \bar{u}_{i|k}, \\
	& \quad \bar{u}_{i|k} = K e_{i|k} + \mu_{i|k}, \; i \in \mathbb{Z}_{[0,N-1]},\\
	& \quad \bar{u}_{N|k} = K_f e_{N|k},
	\end{align}
\end{subequations}
and \eqref{eq:vertex tube constraints-1}-\eqref{eq:vertex tube constraints-3}, \eqref{eq:vertex state input constraint}, \eqref{eq:vertex terminal constraint-1}-\eqref{eq:vertex terminal constraint-4}, with
\begin{align}\label{eq:stage cost function}
    \ell \left( e_{i|k},\bar{u}_{i|k} \right) = \| e_{i|k} \|_{Q}^2  + \| \bar{u}_{i|k} \|_{R}^2,
\end{align}
where $Q \in \mathbb{S}^{4}_{\succeq 0}$ is the weighting matrix for penalizing velocity errors and $R \in \mathbb{S}^{2}_{\succ 0}$ is the weighting matrix for control input penalization. With the offline computed $K_{f}$, $P(\hat{\theta}(k))$ is adapted by using the following condition with online updating $\hat{\theta}(k)$ by \eqref{eq:proj_theta}:
\begin{equation}\label{eq:update P_theta}
 P(\hat{\theta}(k)) - A_{c}(\hat{\theta}(k))^{\top}P(\hat{\theta}(k))A_{c}(\hat{\theta}(k)) - Q - K_{f}^{\top}RK_{f} \succeq 0,
\end{equation}
where $A_{c}(\hat{\theta}(k)) = A(\hat{\theta}(k))+B(\hat{\theta}(k))K_{f}$.

After solving \eqref{problem:RAMPC} at each sampling time instant $k \in \mathbb{Z}_{0+}$, the control action is chosen as 
\begin{align}
    u(k) = K (x(k)-x_s(k)) + \mu^*_{0|k}+ u_s(k),
\end{align}
where $\mu^{*}_{0|k}$ is the the first element of the optimal solution.

\subsection{Closed-loop Property Analysis}\label{sec:property_closed}

We next discuss the closed-loop properties of the wheelchair operated by the proposed safety-based speed tracking MPC controller. The theoretical results are summarized in the following theorem.

\begin{theorem}[Closed-loop Properties]\label{theorem:closed-loop property}
    Consider Assumptions \ref{assumption:bounded disturbances}-\ref{assumption:Kf} hold. Given a speed reference $v_{d}$ and a feasible initial state $x(0)$, the wheelchair system \eqref{eq:sys_dis_vary} operated by the proposed robust adaptive tracking MPC in \eqref{problem:RAMPC} and estimated parameter updated by \eqref{eq:proj_theta} is recursively feasible and input-to-state stable (ISS) to a neighbourhood of admissible steady state $x_s$ obtained from the optimization problem \eqref{prob:xs_us} and the neighbourhood is defined with the bounds of uncertainty parameter $\theta$ and external disturbance $w$.
\end{theorem}
    
\begin{proof}
    \textbf{(Recursive Feasibility)} Consider the MPC problem~\eqref{problem:RAMPC} is feasible at a sampling time $k \in \mathbb{Z}_{0+}$. The optimal solution of \eqref{problem:RAMPC} at time $k$ can be denoted by 
    \begin{align*}
    \bm{\mu}^*(k) &= \{\mu_{0|k}^*, \ldots, \mu_{N-1|k}^{*} \},\\
    \bm{\alpha}^*(k) &= \{\alpha_{0|k}^*, \ldots, \alpha_{N|k}^{*}\},\\
    \bm{\Omega}^*(k) &= \{\Omega_{j,0|k}^*, \ldots, \Omega_{j,N|k}^*, \; j \in \mathbb{Z}_{[1,p]} \}, 
    \end{align*}
    from which, at the next time instant $k+1$, a suboptimal solution sequence can be constructed as follows: $\bm{\mu}(k+1) = \{\mu_{1|k}^{*},\allowbreak \ldots, \mu_{N-1|k}^{*}, \mu_{N-1|k+1} \}$, $\bm{\alpha}(k+1) = \{\alpha_{1|k}^*, \ldots, \alpha_{N|k}^{*}, \alpha_{N|k}^{*}\} $, and $\bm{\Omega}(k+1) = \{\Omega_{j,1|k}^*, \ldots, \Omega_{j,N|k}^*, \Omega_{j,N|k}^*, \; j \in \mathbb{Z}_{[1,p]} \} $. Note that from the conditions discussed in Assumptions \ref{assumption:K}-\ref{assumption:Kf}, there exists $\mu_{N-1|k+1}$ such that
    \begin{align*}
        K S_j \alpha^*_{N|k} + \mu_{N-1|k+1} = K_f S_j \alpha^*_{N|k}, \, j \in \mathbb{Z}_{[1,p]}.
    \end{align*}
    
    As a special case mentioned in Remark \ref{remark:Kf=K}, if $K_f$ is chosen to be the same as $ K $, the above condition holds with $\mu_{N-1|k+1} = \mathbf{0}_{2,1}$. The above suboptimal solution at time $k+1$ satisfies all the constraints of \eqref{problem:RAMPC} for $\Theta(k+1) \subseteq \Theta(k)$ indicated by \eqref{eq:theta set update}. Therefore, the MPC optimization problem \eqref{problem:RAMPC} is also feasible at the next time $k+1$. Thus, it is recursively feasible.

    \textbf{(ISS Stability)} 
    Since the closed-loop system is recursive feasible, the Lyapunov candidate function for this closed-loop system is chosen as the optimal MPC cost function that can be denoted by $V^*(x(k),\hat{\theta}(k),x_s)$. 
    
    From the MPC cost function defined in \eqref{eq:RAMPC-cost function}, we know
    \begin{align}\label{eq:Lyapunov V-1-1}
        V^*(x(k),\hat{\theta}(k),x_s) &\geq \| e_{0|k} \|_{Q}^2= \|x(k)-x_s\|_Q^2 \nonumber\\
        &\geq \underline{\lambda}(Q) \|x(k)-x_s\|^2 \nonumber\\
        & \triangleq \sigma_1 \left( \|x(k)-x_s\| \right),
    \end{align}
    where $\sigma_1 (\cdot)$ is a $ \mathcal{K}_{\infty}$ function. In addition, for quadratic stage cost functions, it can be verified that there exists a $ \mathcal{K}_{\infty}$ function $\sigma_2(\cdot)$ such that \cite{rawlings2017model} 
    \begin{align}\label{eq:Lyapunov V-1-2}
         V^*(x(k),\hat{\theta}(k),x_s) \leq \sigma_2 \left( \| x(k)-x_s \| \right).
    \end{align}
    
    Then, we can derive
    \begin{align*}
        & \; V^*(x(k+1),\hat{\theta}(k+1),x_s) - V^*(x(k),\hat{\theta}(k),x_s)\\
        \leq  & \; V(x(k+1),\hat{\theta}(k+1),x_s) - V^*(x(k),\hat{\theta}(k),x_s),
    \end{align*}
    where $V(x(k+1),\hat{\theta}(k+1),x_s)$ denotes an MPC cost function at time $k+1$ with a suboptimal solution.
    
    From the optimal solution at time instant $k$, a suboptimal input sequence at time $k+1$ can be chosen to be
    \begin{align*}
        \bar{u}_{i|k+1} &= K e_{i|k+1} + \mu^*_{i+1|k}, \; i \in \mathbb{Z}_{[0,N-2]},\\
        \bar{u}_{N-1|k+1} &= K_f e_{N-1|k+1},\\
        \bar{u}_{N|k+1} &= K_f e_{N|k+1}.
    \end{align*}
    
    With this suboptimal input sequence, we have that 
    \begin{align*}
        & \;V(x(k+1),\hat{\theta}(k+1),x_s) - V^*(x(k),\hat{\theta}(k),x_s) \\
        =  & \; \sum_{i=0}^{N-2}\left( \ell \left( e_{i|k+1},\bar{u}_{i|k+1} \right) -\ell \left( e^*_{i+1|k},\bar{u}^*_{i+1|k} \right) \right)   \\
        & \;\; - \ell ( e^*_{0|k},\bar{u}^*_{0|k} ) + \ell \left( e_{N-1|k+1},\bar{u}_{N-1|k+1} \right) \\
        & \;\; + \| e_{N|k+1} \|_{P(\hat{\theta}(k+1))}^2 - \| e_{N-1|k+1} \|_{P(\hat{\theta}(k+1))}^2\\
        & \;\; + \| e_{N-1|k+1} \|_{P(\hat{\theta}(k+1))}^2 - \| e^*_{N|k} \|_{P(\hat{\theta}(k))}^2.
    \end{align*}
    
    From \eqref{eq:update P_theta} at next time instant $k+1$, it is obvious that
    \begin{align*}
        \ell ( e_{N-1|k+1}, & \bar{u}_{N-1|k+1} )  + \| e^*_{N|k+1} \|_{P(\hat{\theta}(k+1))}^2 \\
        & - \| e_{N-1|k+1} \|_{P(\hat{\theta}(k+1))}^2 \leq 0,
    \end{align*}
    which leads to
    \begin{align*}
        & \;V(x(k+1),\hat{\theta}(k+1),x_s) - V^*(x(k),\hat{\theta}(k),x_s) \\
        \leq  & \; \sum_{i=0}^{N-2}\left( \ell \left( e_{i|k+1},\bar{u}_{i|k+1} \right) -\ell \left( e^*_{i+1|k},\bar{u}^*_{i+1|k} \right) \right)   \\
        & \;\; - \ell ( e^*_{0|k},\bar{u}^*_{0|k} ) + \| e_{N-1|k+1} \|_{P(\hat{\theta}(k+1))}^2 - \| e^*_{N|k} \|_{P(\hat{\theta}(k))}^2.
    \end{align*}
    
    From the parameter update rule indicated by the optimization problem \eqref{eq:proj_theta}, it yields
    \begin{align*}
        \|\hat{\theta}(k+1) - \hat{\theta}(k)\| \leq \|\hat{\theta}(k) - \theta^* \|. 
    \end{align*}

    Then, for each $i \in \mathbb{Z}_{[0,N-2]}$, there exist $\mathcal{K}_{\infty}$ functions $\varrho_i(\cdot)$, $\eta_i (\cdot)$ and $\zeta_i(\cdot)$ such that
    \begin{align*}
        & \; \ell \left( e_{i|k+1},\bar{u}_{i|k+1} \right) -\ell \left( e^*_{i+1|k},\bar{u}^*_{i+1|k} \right)\\
        \leq & \; \varrho_i (\|e_{i|k+1}  -e^*_{i+1|k} \|)\\
        \leq & \; \eta_i (\|w\|) + \zeta_i (\|\hat{\theta}(k) - \theta^*\|), \, i \in \mathbb{Z}_{[0,N-2]}.
    \end{align*}
    
    Furthermore, for quadratic functions, there exist $\mathcal{K}_{\infty}$ functions $\varrho_{N-1}(\cdot)$, $\eta_{N-1} (\cdot)$, $\zeta_{N-1}(\cdot)$ and $\tau(\cdot)$ such that 
    \begin{align*}
        & \; \| e_{N-1|k+1} \|_{P(\hat{\theta}(k+1))}^2 - \| e^*_{N|k} \|_{P(\hat{\theta}(k))}^2 \\
        \leq  & \; \varrho_{N-1} (\|e_{N-1|k+1}  -e^*_{N|k} \|) + \tau (\|\hat{\theta}(k+1) - \hat{\theta}(k)\|)\\
        \leq  & \; \eta_{N-1} (\|w\|) + \zeta_{N-1} (\|\hat{\theta}(k) - \theta^*\|) + \tau (\|\hat{\theta}(k) - \theta^*\|).
    \end{align*}
    
    From the definition of $\ell ( e^*_{0|k},\bar{u}^*_{0|k} )$ in \eqref{eq:stage cost function}, there exists a $\mathcal{K}_{\infty}$ function $\sigma_3(\cdot)$ such that
    \begin{align*}
        \ell ( e^*_{0|k},\bar{u}^*_{0|k} ) \geq  \sigma_3 ( \| e^*_{0|k}\| )= \sigma_3 \left( \|x(k)-x_s\| \right).
    \end{align*}
    
    Finally, by combining the above derived bounds for each term, we can conclude that
    \begin{align}\label{eq:Lyapunov V-2}
        & \; V^*(x(k+1),\hat{\theta}(k+1),x_s) - V^*(x(k),\hat{\theta}(k),x_s) \nonumber\\
        \leq  &  - \sigma_3 \left( \|x(k)-x_s\| \right) + \eta_w \left( \|w\| \right) + \zeta_{\theta} ( \| \hat{\theta}(k) - \theta^* \|),
    \end{align}
    with
    \begin{align*}
        \eta_w \left( \|w\| \right) &= \sum_{i=0}^{N-1} \eta_{i} (\|w\|),\\
        \zeta_{\theta} ( \| \hat{\theta}(k) - \theta^* \|) &= \sum_{i=0}^{N-1}\zeta_i (\|\hat{\theta}(k) - \theta^*\|) + \tau (\|\hat{\theta}(k) - \theta^*\|).\\
    \end{align*}
    
    Based on \cite[Definition 7 and Remark 5]{limon2009input}, the chosen Lyapunov function satisfies the conditions in \eqref{eq:Lyapunov V-1-1}, \eqref{eq:Lyapunov V-1-2} and \eqref{eq:Lyapunov V-2}, which proves that it is an ISS Lyapunov function. Then, the wheelchair system \eqref{eq:sys_dis_vary} operated by the proposed robust adaptive tracking MPC in \eqref{problem:RAMPC} and estimated parameter updated by \eqref{eq:proj_theta} is ISS to a neighbourhood of $x_s$ being close to the given speed reference $v_d$.
\end{proof}

\section{Results}\label{sec:results}

In this section, we apply the proposed control algorithm into an electric-powered wheelchair. The experiment with the wheelchair has been carried out to collect voltage and velocity data, which are used for system identification with parameter estimation. The results with a high-fidelity wheelchair model show the effectiveness of the proposed control algorithm.

\subsection{Experiment Setup}

To validate the effectiveness of the proposed method, the speed tracking control is implemented on a high-fidelity model of the wheelchair as shown in Fig.~\ref{fig:wheelchair_photo}. This wheelchair is powered by a battery with maximum voltage of $24$ V. The two wheels are driven by Brushless DC motors with a maximum velocity of $2$ m/s. An absolute rotary encoder with 4096 pulses per revolution is used for velocity measurement.

\begin{figure}[thbp]
	\centerline{\includegraphics[width = \columnwidth]{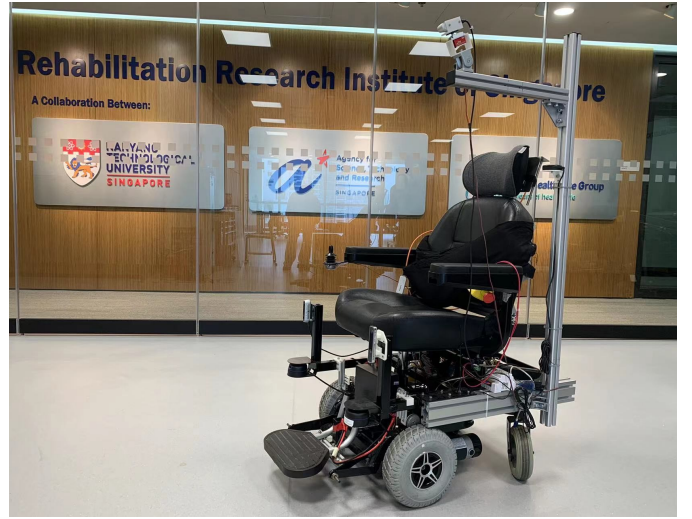}}
	\caption{Electric powered wheelchair in experiment.}{\label{fig:wheelchair_photo}}
\end{figure}	
	
An experiment participant with $80$ kg drives the wheelchair with open-loop voltage input using the joystick to fully excite the system for system identification. Based on the measured voltage input and velocity output, the system parameters are identified using the least-square method and the model validation is conducted by comparing the actual velocity measurement with the estimated velocity based on the proposed model. The identified motor parameters are: $l_{1} = 0.614$, $l_{2} = 0.482$, $r_{1} = 8.138$, $r_{2} = 5.871$, $k_{e_1} = 3.471$, $k_{e_2} = 2.610$, $k_{t_1} = 1.324$ and $k_{t_2} = 0.827$. A model validation result for right wheel is shown in Fig.~\ref{fig:model_vali}. It can be seen that with the same voltage in model validation, the estimated velocity follows the trend of the measured velocity.

\begin{figure}[tbp]
    \subfigure[]{
    \includegraphics[width = \columnwidth]{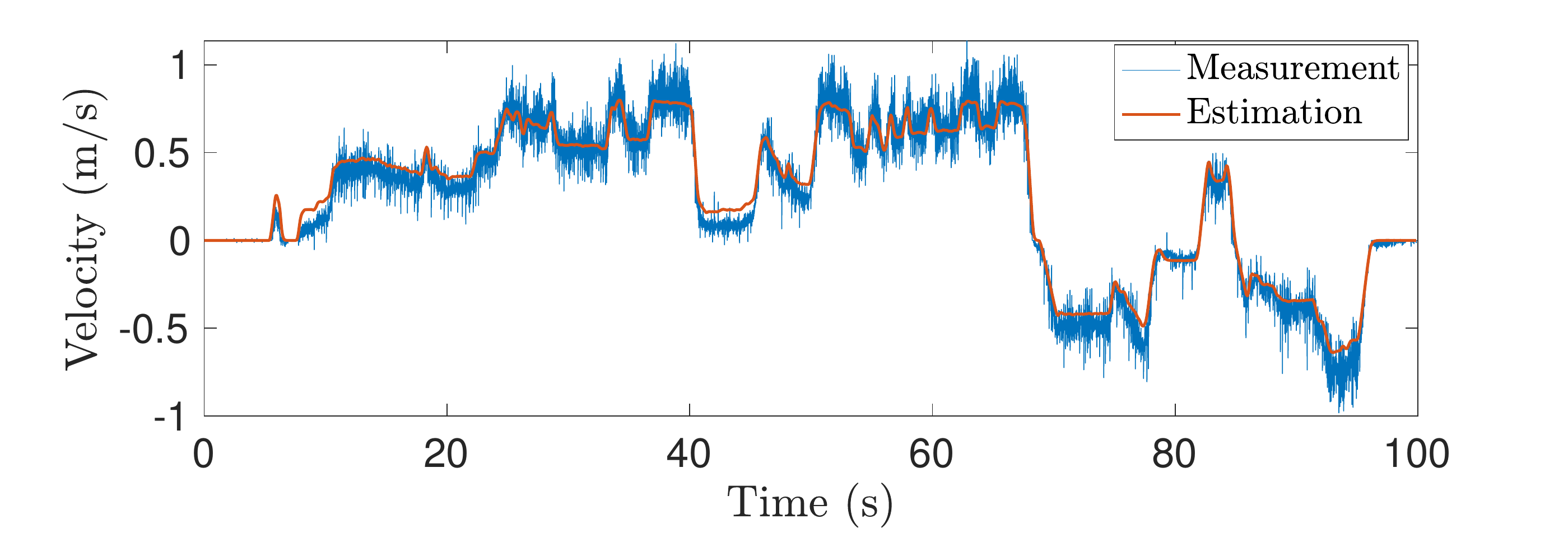}\label{fig:model_vali_a}
    }
    \subfigure[]{
    \includegraphics[width = \columnwidth]{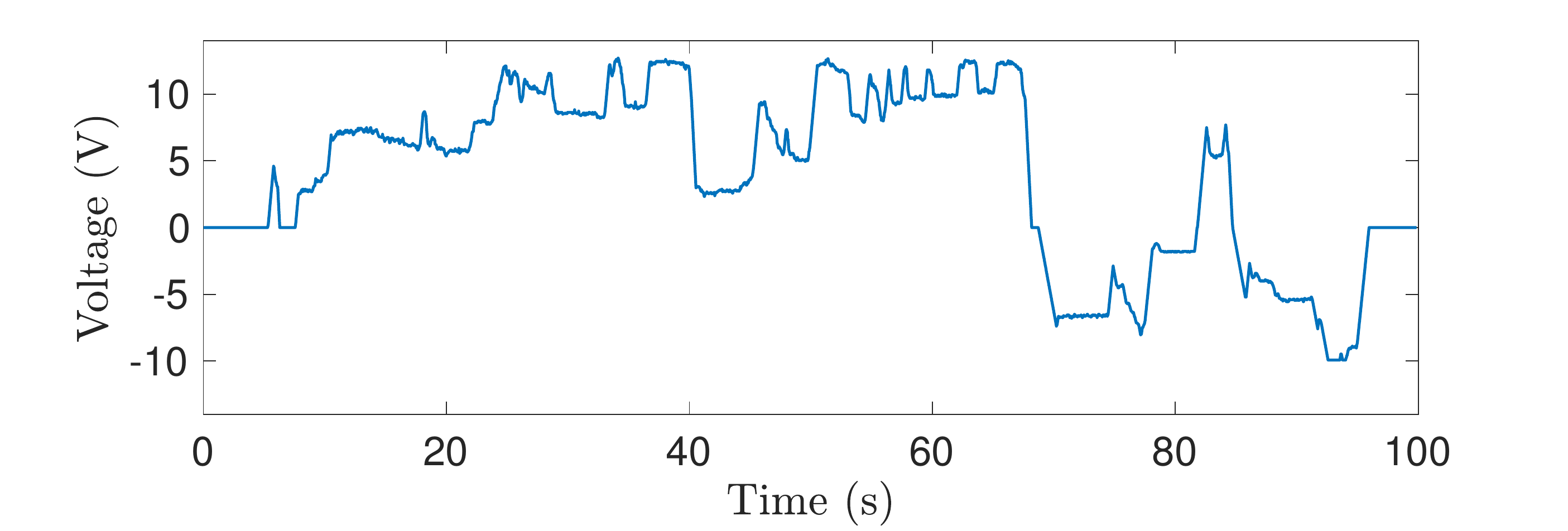}\label{fig:model_vali_b}
    }
    \caption{Model validation result: (a) measured and estimated velocity of right wheel; (b) input voltage.}
    \label{fig:model_vali}
\end{figure}

\subsection{Parameters in the Wheelchair LPV Model}

The parameter uncertainty exists when the wheelchair user changes, i.e., $M^{-1} = \bar{M}^{-1}+\beta\tilde{M}$ and $D = \bar{D} +\gamma \tilde{D}$, where $\bar{M}^{-1}$ and $\bar{D}$ are the nominal matrix of $M^{-1}$ and $D$. The uncertain parameter becomes $\theta = \left[\begin{array}{ccc}
\beta & \gamma & \beta\gamma\end{array}\right]^{\top}$ for $q=3$ and the system model matrices in \eqref{eq:sys_para_variation} can be desribed as follows:
\begin{align*}
	A_{0} &=\left[\begin{array}{cc}
		I_{2} & T_{s}I_{2}\\
		-T_{s}\bar{\Gamma} & I_{2}-T_{s}(\bar{M}^{-1}\bar{D}+L^{-1}R)
	\end{array}\right],\\
	A_{1}& =\left[\begin{array}{cc}
		\mathbf{0}_{2,2} & \mathbf{0}_{2,2}\\
		-T_{s}(\tilde{M}L^{-1}R\bar{D}+\tilde{M}K_{t}L^{-1}K_{e}) & -T_{s}\tilde{M}\bar{D}
	\end{array}\right],\\
	A_{2} &=\left[\begin{array}{cc}
		\mathbf{0}_{2,2} & \mathbf{0}_{2,2}\\
		-T_{s}\bar{M}^{-1}L^{-1}R\tilde{D} & -T_{s}\bar{M}^{-1}\tilde{D}
	\end{array}\right],\\
	A_{3} &=\left[\begin{array}{cc}
		\mathbf{0}_{2,2} & \mathbf{0}_{2,2}\\
		-T_{s}\tilde{M}L^{-1}R\tilde{D} & -T_{s}\tilde{M}\tilde{D}
	\end{array}\right],\\
	B_{0}& =\left[\begin{array}{c}
		\mathbf{0}_{2,2}\\
		T_{s}\bar{M}^{-1}K_{t}L^{-1},
	\end{array}\right],\; B_{1}=\left[\begin{array}{c}
		\mathbf{0}_{2,2}\\
		T_{s}\tilde{M}K_{t}L^{-1}
	\end{array}\right],\\
	B_{2}&=\mathbf{0}_{4,2}, \;B_{3}=\mathbf{0}_{4,2},
\end{align*}
where $\bar{\Gamma} = \bar{M}^{-1}L^{-1}R\bar{D}+\bar{M}^{-1}K_{t}L^{-1}K_{e}$. By first fixing the identified value of motor parameters, the values of nominal matrices $\bar{M}^{-1}$, $\bar{D}$ and the difference matrices $\tilde{M}$, $\tilde{D}$ are computed based on three different participants with weight range from $60$ kg to $85$ kg with detailed value as follows:
\begin{align*}
    \bar{M}^{-1} &=\left[\begin{array}{cc}
    3.822 & 1.776\\
    1.549 & 4.839
    \end{array}\right],\, \bar{D} = \mathrm{diag}(2.034,1.854),\\
    \tilde{M} &= \left[\begin{array}{cc}
    1.241 & -0.046\\
    0.733 & 0.185
    \end{array}\right], \tilde{D} = \mathrm{diag}(0.048, 0.024).
\end{align*}

The uncertain parameter $\theta$ is considered within the constraint $ [-0.8, -0.8, -0.64]^\top \leq \theta \leq [0.8, 0.8, 0.64]^\top $. A participant with true system parameters $\theta^{*} = [0.5, 0.6, 0.3]^{\top}$ is involved in the speed tracking. The initial parameter estimation is set as $\hat{\theta}(0) = [0.55,0.1,0.055]^\top$.

\subsection{Speed Tracking Results}

We next validate the performance of the proposed control algorithm on the wheelchair using two speed profiles (Task A and Task B). As shown in Fig.~\ref{fig:wheelchair_move}, the first case (Task A) is climbing and descending inclined ramps. This task is similar to the task in Cybathlon event \cite{ETHZurich2020}, but with a smaller maximum slope angle as the design of daily life ramps considers the capacity of different types of wheelchairs, and the maximum slope for hand-propelled wheelchairs is $5$ degree \cite{BrainLine2008}. This task involves point-to-point movement and is representative since inclined surfaces are common in daily life. Here, the reference of velocity and acceleration are generated by an available planner where the constraints of velocity and acceleration are considered at the planning level. According to Assumption \ref{assumption:bounded disturbances}, the disturbance torques are unknown but within the constraint $-[6.5,6.5]^{\top}\leq w \leq [6.5,6.5]^{\top}$ Nm, and the gravity force of operator introduces a constant disturbance on the system when the wheelchair is on the ramps.

\begin{figure}[tb]
	\centerline{\includegraphics[width = 0.85\columnwidth]{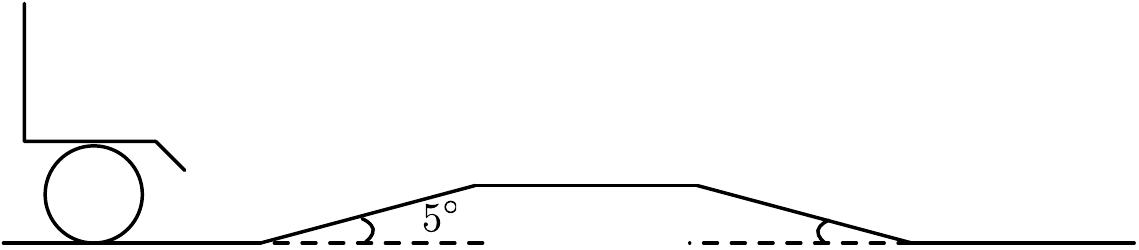}}
	\caption{Task of climbing and descending inclined ramps for wheelchair.}{\label{fig:wheelchair_move}}
\end{figure}

To ensure the safety of wheelchair users while providing a comfortable transient performance, the wheelchair is required to operate within $1$ m/s velocity constraint while respecting $2$ m/s$^2$ maximum acceleration constraint \cite{Gulati2008}. The speed on the flat surface, climbing and descending ramps are set as $1$ m/s, $0.9$ m/s and $0.5$ m/s, respectively. The sampling time of controller is $T_{s} = 0.005$ s. The prediction horizon of MPC is set as $N = 3$. The tuning parameters are set as $Q = \mathrm{diag}(1000, 1000, 1000, 1000)$ and $R = \mathrm{diag}(0.001,0.001)$. For comparison, a proportional-integral (PI) controller is used as a benchmark controller in the performance comparison. 

\begin{figure}[tb]
	\centerline{\includegraphics[width = \columnwidth]{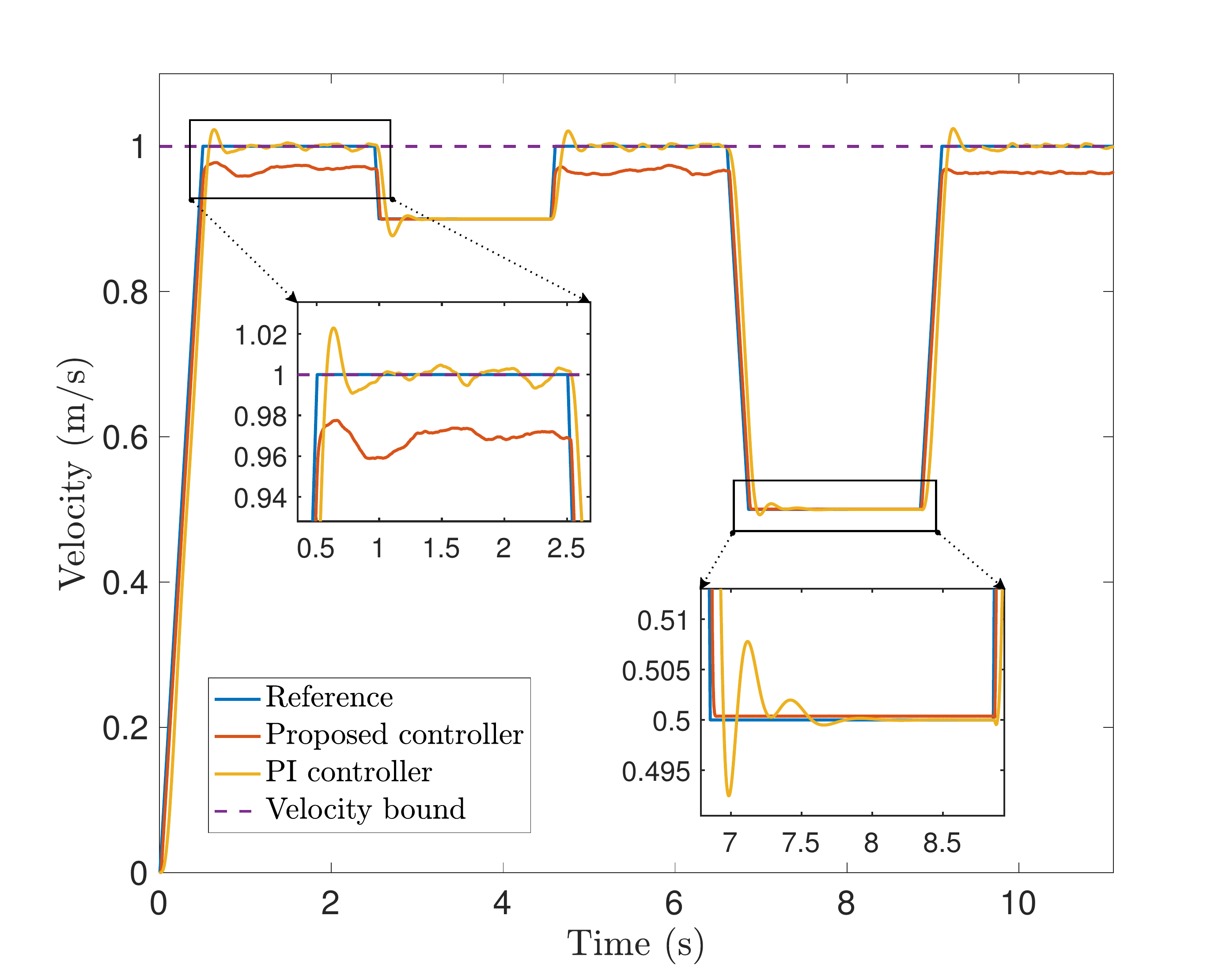}}
	\caption{Speed tracking result of the right wheel of wheelchair in Task A.}{\label{fig:vr_ramp}}
\end{figure}

\begin{figure}[tb]
	\centerline{\includegraphics[width = \columnwidth]{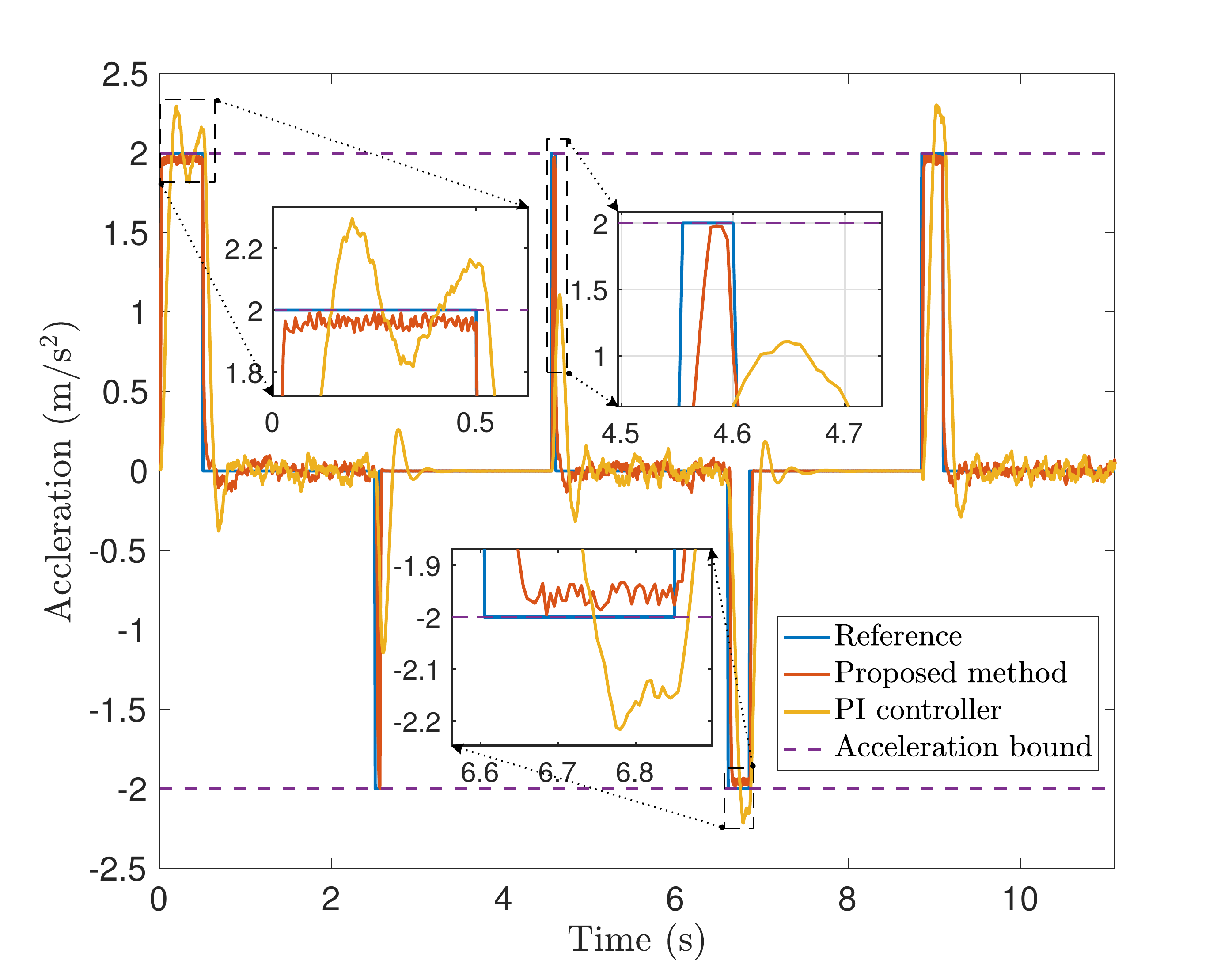}}
	\caption{Acceleration of the right wheel of wheelchair in Task A.}{\label{fig:ar_ramp}}
\end{figure}

Since the reference for the left and right wheels are same in the Task A, only the tracking results of the right wheel are presented and analysed here. The velocity and acceleration of the right wheel based on the proposed method and PI controller for wheelchair riding on ramps are demonstrated in Fig.~\ref{fig:vr_ramp} and Fig.~\ref{fig:ar_ramp}, respectively. 

In the application of wheelchair speed tracking, the speed and acceleration constraints should always be satisfied for the safe operation. However, as shown in Fig.~\ref{fig:vr_ramp} and Fig.~\ref{fig:ar_ramp}, although the reference is generated considering velocity and acceleration constraint satisfaction, with the PI controller, the actual speed and acceleration of wheelchair both violate the given safety constraints. Conversely, the proposed MPC controller provides smoother acceleration and deceleration while keeping the maximum acceleration within tolerance during the transient process. Specifically, during $4.5$ s to $4.6$ s in Fig.~\ref{fig:ar_ramp}, it can be seen that the designed controller tries to push the actuator close to its maximum acceleration while the PI controller only runs at $1$ m/s$^{2}$ maximum acceleration when the wheelchair is moving from climbing ramps to flat ground. During $6.5$ s to $7$ s, the wheelchair is moving from the flat surface to descending ramps. Based on the magnified result in Fig.~\ref{fig:vr_ramp} and Fig.~\ref{fig:ar_ramp}, the proposed method accelerates and decelerates the wheelchair more smoothly without violating the acceleration constraints, whereas the acceleration exceeds the tolerance and introduces velocity oscillation by using the PI controller.

The maximum speed error of the right and left wheels are summarised in Table~\ref{tab:err_ramp}. It can also be seen that the proposed method achieves a smaller speed tracking error on both left and right wheel sides. The input voltage of the right and left wheels are shown in Fig.~\ref{fig:input_ramp} and both the voltage constraints are satisfied for the left and right wheels. Furthermore, the evolution of estimated uncertain parameters are shown in Fig.~\ref{fig:theta}. The true values of uncertain parameters are identified from a participant, which is unknown to the controller. The estimated value tends to the true value during the process.

\begin{figure}[tb]
	\centerline{\includegraphics[width = \columnwidth]{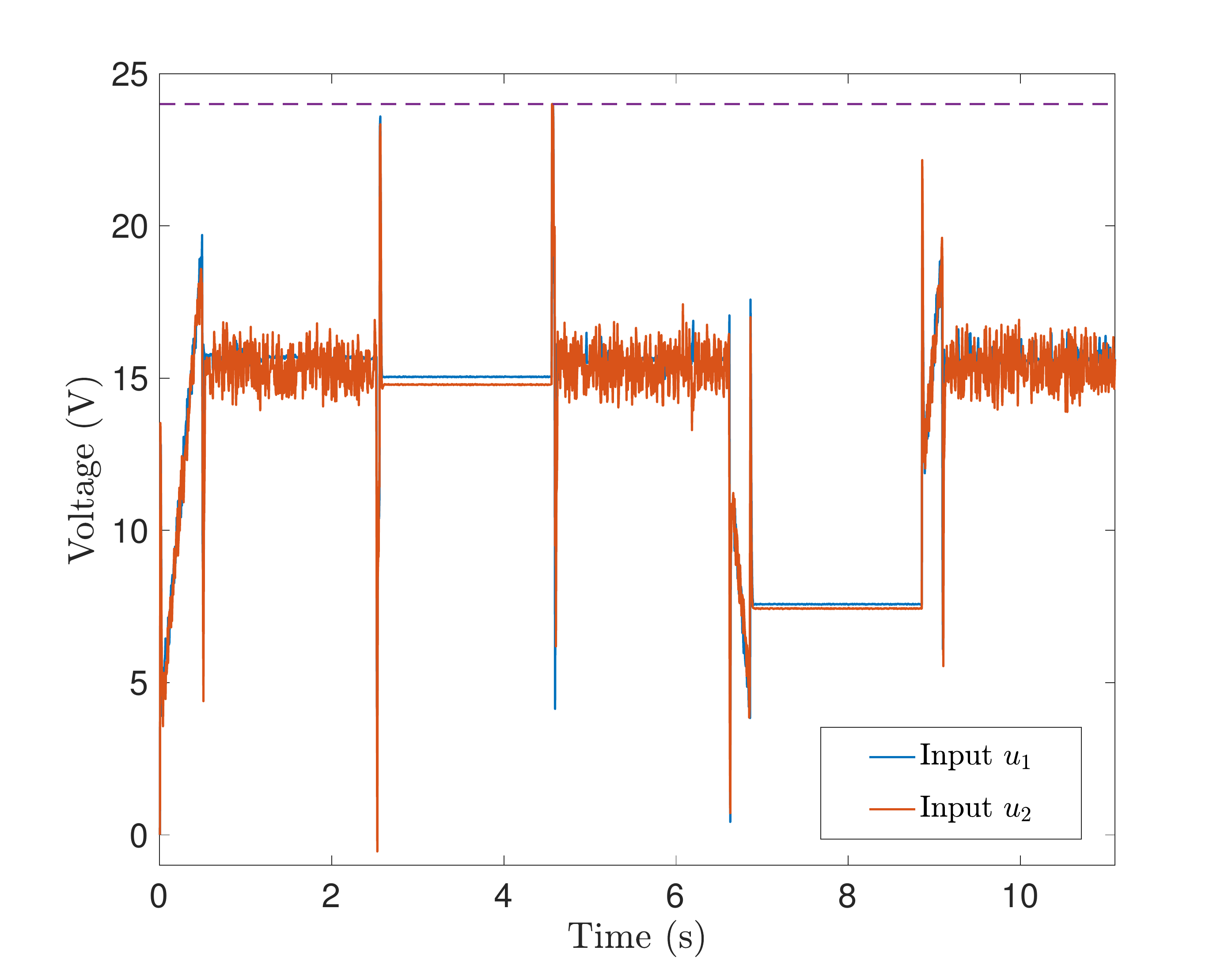}}
	\caption{Input voltage for speed tracking of wheelchairs in Task A.}{\label{fig:input_ramp}}
\end{figure}

\begin{figure}[tb]
	\centerline{\includegraphics[width = \columnwidth]{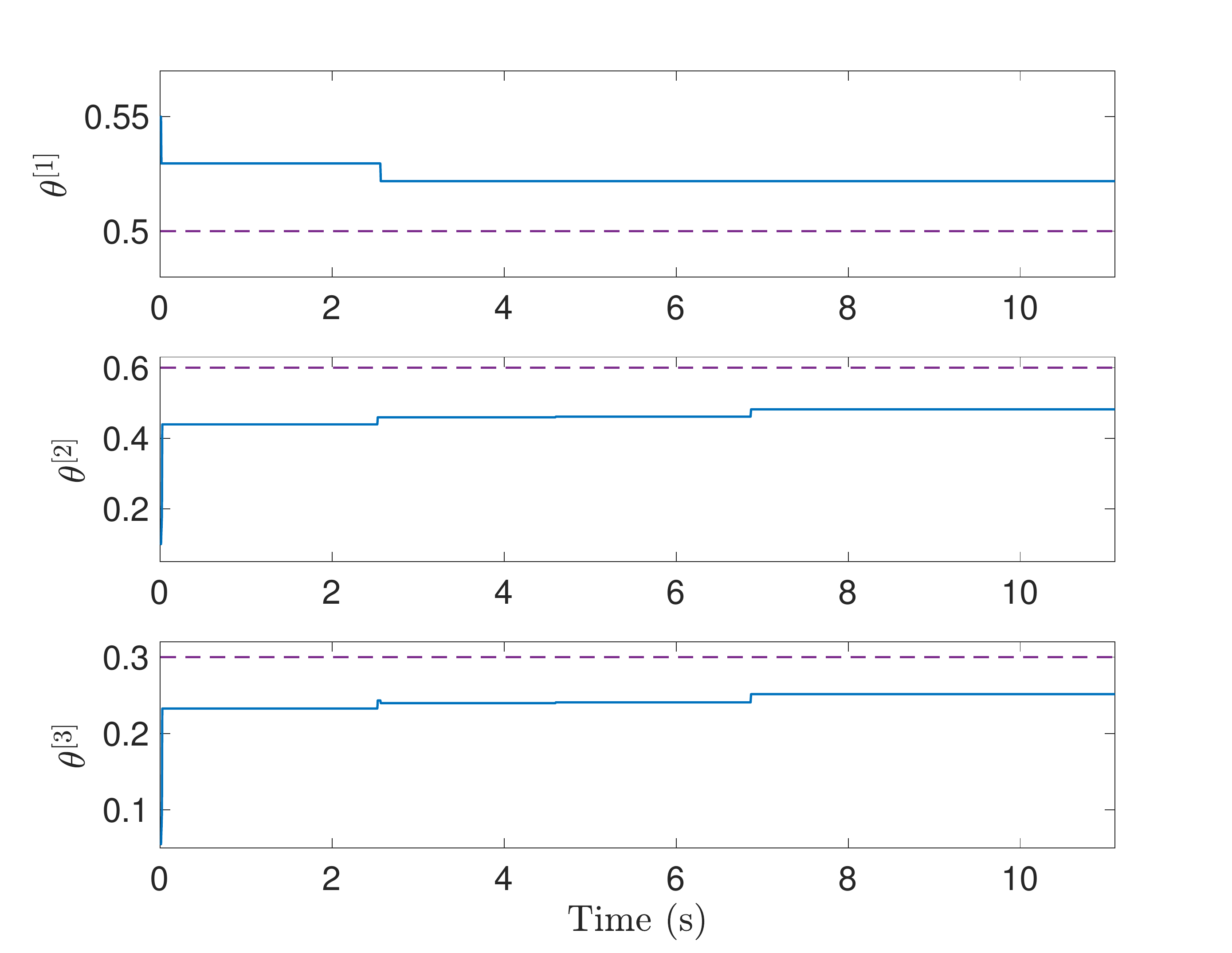}}
	\caption{Evolution of uncertain parameter.}{\label{fig:theta}}
\end{figure}

\begin{table}[t]
\centering
\begin{threeparttable}
\caption{Maximum speed tracking error for the wheelchair in Task A.}\label{tab:err_ramp}
\begin{centering}
\begin{tabular}{ccc}
\hline 
 & Proposed method & PI controller\tabularnewline
\hline 
$\left\Vert e^{[1]}\right\Vert _{\infty}$ & 0.047 & 0.158\tabularnewline
$\left\Vert e^{[2]}\right\Vert _{\infty}$ & 0.048 & 0.149\tabularnewline
\hline 
\end{tabular}
\begin{tablenotes}
	\small
	\item Note: the unit of numbers is m/s.
\end{tablenotes}
	\par\end{centering}
\end{threeparttable}
\end{table}

In Task B, the wheelchair is required to track a sinusoidal reference while respecting maximum velocity constraint at $0.2$ m/s and maximum acceleration at $1$ m/s$^{2}$. The speed references for left and right wheel are $v_{d_1}(t) = -0.5 \sin(2t)$ and $v_{d_2}(t) = 0.5\sin(2t)$, respectively. This is relevant to a practical scenario such as manual manipulating of the wheelchair when an inadmissible reference is given by the users through a joystick, the speed control should be able to follow the reference as close as possible while satisfying the safety constraints classified by velocity and acceleration tolerance. 

The corresponding speed and acceleration during the tracking process are demonstrated in Fig.~\ref{fig:v_sin} and Fig.~\ref{fig:a_sin}, respectively. It can be seen from Fig.~\ref{fig:v_sin} that, when the reference speeds are admissible, the actual speeds of the left and right wheels follow the reference closely. The proposed control algorithm steers the speed of wheelchair to the closed admissible reference instead during $0.2$ s to $1.4$ s and $1.8$ s to $2.9$ s when the reference is inadmissible. Meanwhile, it can also be seen from Fig.~\ref{fig:a_sin} that the proposed controller ensures the acceleration constraint during the whole process. 

\begin{figure}[tb]
	\centerline{\includegraphics[width = \columnwidth]{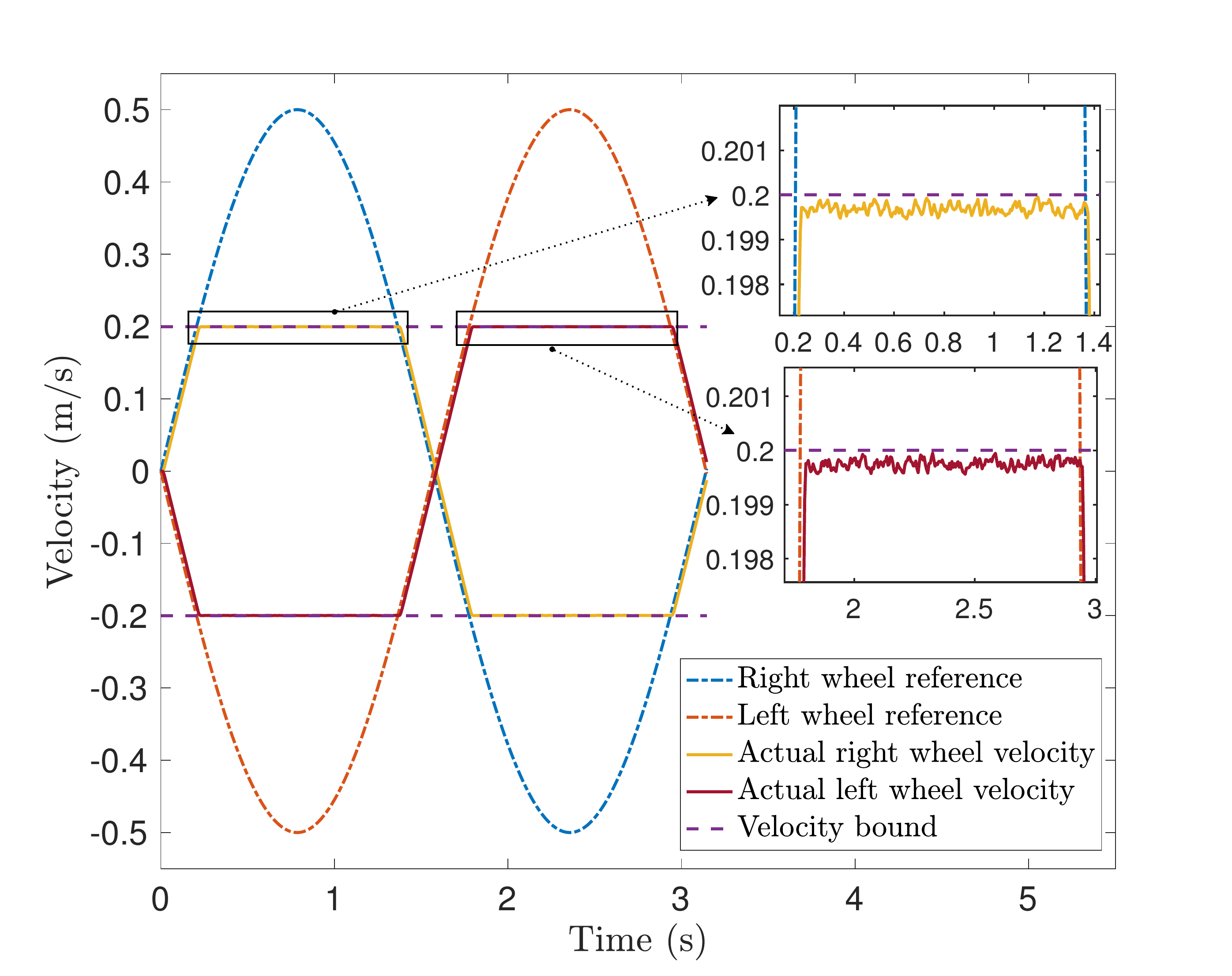}}
	\caption{Speed result for sinusoidal reference tracking.}{\label{fig:v_sin}}
\end{figure}

\begin{figure}[t]
	\centerline{\includegraphics[width = \columnwidth]{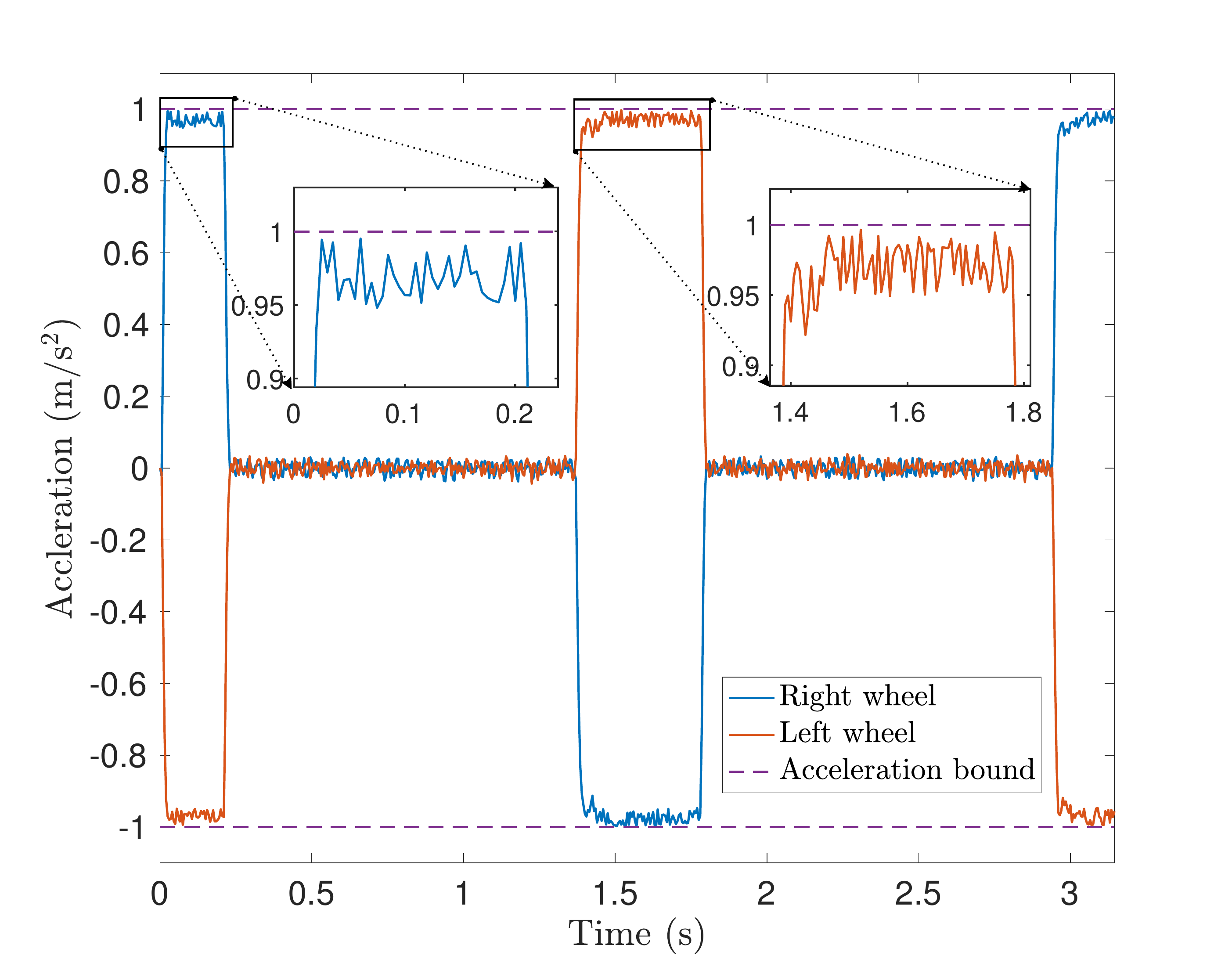}}
	\caption{Acceleration result for sinusoidal reference tracking.}{\label{fig:a_sin}}
\end{figure}

\section{Conclusion}\label{sec:conclusion}

In this work, we have presented a robust adaptive MPC algorithm for achieving safety-based speed tracking of an electric powered wheelchair. The external disturbance introduced in different scenarios and varying parameters due to change of users have been explicitly considered in the controller design. To guarantee recursive feasibility and ensure constraint satisfaction in the proposed MPC scheme, a set-membership approach has been utilized for online set propagation bounding unknown but constant uncertain parameters. The pair of steady state and input has been introduced when a given reference is inadmissible. For the considered LPV system with parameter uncertainties of the wheelchair, the state and input constraints can be robustly guaranteed. The effectiveness of the proposed algorithm has been shown by two representative task results on a wheelchair. The developed algorithm is shown to achieve decent tracking results while ensuring the state and input constraints for admissible and inadmissible references when different users and different application scenarios are involved.

\bibliographystyle{IEEEtranDOI}
\bibliography{Arxiv_file}	

\begin{thebibliography}{10}
\providecommand{\url}[1]{#1}
\csname url@samestyle\endcsname
\providecommand{\newblock}{\relax}
\providecommand{\bibinfo}[2]{#2}
\providecommand{\BIBentrySTDinterwordspacing}{\spaceskip=0pt\relax}
\providecommand{\BIBentryALTinterwordstretchfactor}{4}
\providecommand{\BIBentryALTinterwordspacing}{\spaceskip=\fontdimen2\font plus
\BIBentryALTinterwordstretchfactor\fontdimen3\font minus
  \fontdimen4\font\relax}
\providecommand{\BIBforeignlanguage}[2]{{%
\expandafter\ifx\csname l@#1\endcsname\relax
\typeout{** WARNING: IEEEtran.bst: No hyphenation pattern has been}%
\typeout{** loaded for the language `#1'. Using the pattern for}%
\typeout{** the default language instead.}%
\else
\language=\csname l@#1\endcsname
\fi
#2}}
\providecommand{\BIBdecl}{\relax}
\BIBdecl

\bibitem{Haisma2006}
J.~A. Haisma, L.~H. {Van Der Woude}, H.~J. Stam, M.~P. Bergen, T.~A. Sluis, and
  J.~B. Bussmann, ``{Physical capacity in wheelchair-dependent persons with a
  spinal cord injury: A critical review of the literature},'' \emph{Spinal
  Cord}, vol.~44, no.~11, pp. 642--652, Nov. 2006, doi: 10.1038/sj.sc.3101915.

\bibitem{Jang2016}
G.~Jang, J.~Kim, S.~Lee, and Y.~Choi, ``Emg-based continuous control scheme
  with simple classifier for electric-powered wheelchair,'' \emph{IEEE
  Transactions on Industrial Electronics}, vol.~63, no.~6, pp. 3695--3705, Jun.
  2016, doi: 10.1109/TIE.2016.2522385.

\bibitem{Candiotti2019}
J.~L. Candiotti, B.~J. Daveler, D.~C. Kamaraj, C.~S. Chung, R.~Cooper, G.~G.
  Grindle, and R.~A. Cooper, ``A heuristic approach to overcome architectural
  barriers using a robotic wheelchair,'' \emph{IEEE Transactions on Neural
  Systems and Rehabilitation Engineering}, vol.~27, no.~9, pp. 1846--1854, Sep.
  2019, doi: 10.1109/TNSRE.2019.2934387.

\bibitem{Fehr2000}
L.~Fehr, W.~E. Langbein, and S.~B. Skaar, ``{Adequacy of power wheelchair
  control interfaces for persons with severe disabilities: A clinical
  survey},'' \emph{Journal of Rehabilitation Research and Development},
  vol.~37, no.~3, pp. 353--360, 2000.

\bibitem{lian2015near}
C.~Lian, X.~Xu, H.~Chen, and H.~He, ``Near-optimal tracking control of mobile
  robots via receding-horizon dual heuristic programming,'' \emph{IEEE
  transactions on cybernetics}, vol.~46, no.~11, pp. 2484--2496, Nov. 2015.

\bibitem{li2017adaptive}
L.~Li, Y.-H. Liu, T.~Jiang, K.~Wang, and M.~Fang, ``Adaptive trajectory
  tracking of nonholonomic mobile robots using vision-based position and
  velocity estimation,'' \emph{IEEE Transactions on Cybernetics}, vol.~48,
  no.~2, pp. 571--582, Feb. 2017.

\bibitem{li2021trajectory}
P.~Li, S.~Wang, H.~Yang, and H.~Zhao, ``Trajectory tracking and obstacle
  avoidance for wheeled mobile robots based on empc with an adaptive prediction
  horizon,'' \emph{IEEE Transactions on Cybernetics}, 2021.

\bibitem{Zhang2014}
Y.~Zhang, G.~Liu, and B.~Luo, ``Finite-time cascaded tracking control approach
  for mobile robots,'' \emph{Information Sciences}, vol. 284, pp. 31--43, Nov.
  2014, doi: 10.1016/j.ins.2014.06.037.

\bibitem{DeLaCruz2011}
C.~{De La Cruz}, T.~F. Bastos, and R.~Carelli, ``{Adaptive motion control law
  of a robotic wheelchair},'' \emph{Control Engineering Practice}, vol.~19,
  no.~2, pp. 113--125, Feb. 2011.

\bibitem{Fu2013}
J.~Fu, T.~Chai, C.~Y. Su, and Y.~Jin, ``Motion/force tracking control of
  nonholonomic mechanical systems via combining cascaded design and
  backstepping,'' \emph{Automatica}, vol.~49, no.~12, pp. 3682--3686, Dec.
  2013, doi: 10.1016/j.automatica.2013.09.004.

\bibitem{Fu2020}
J.~Fu, F.~Tian, T.~Chai, Y.~Jing, Z.~Li, and C.~Y. Su, ``Motion tracking
  control design for a class of nonholonomic mobile robot systems,'' \emph{IEEE
  Transactions on Systems, Man, and Cybernetics: Systems}, vol.~50, no.~6, pp.
  2150--2156, Jun. 2020, doi: 10.1109/TSMC.2018.2804948.

\bibitem{Do2004}
K.~Do, Z.~Jiang, and J.~Pan, ``Simultaneous tracking and stabilization of
  mobile robots: An adaptive approach,'' \emph{IEEE Transactions on Automatic
  Control}, vol.~49, no.~7, pp. 1147--1152, Jul. 2004, doi:
  10.1109/TAC.2004.831139.

\bibitem{Zhecheng}
\BIBentryALTinterwordspacing
Q.~Zhecheng, ``Wheelchairs join sharing economy hype, puzzle chinese
  consumers.'' [Online]. Available:
  \url{https://www.sixthtone.com/news/1002476/wheelchairs-join-sharing-economy-hype\%2C-puzzle-chinese-consumers}
\BIBentrySTDinterwordspacing

\bibitem{Mayne2000}
D.~Q. Mayne, J.~B. Rawlings, C.~V. Rao, and P.~O.~M. Scokaert, ``{Constrained
  model predictive control: Stability and optimality},'' \emph{Automatica},
  vol.~36, no.~6, pp. 789--814, Jun. 2000, doi: 10.1016/S0005-1098(99)00214-9.

\bibitem{Mayne2014}
D.~Q. Mayne, ``Model predictive control: Recent developments and future
  promise,'' \emph{Automatica}, vol.~50, no.~12, pp. 2967--2986, Jun. 2014,
  doi: 10.1016/j.automatica.2014.10.128.

\bibitem{Fleming2015}
J.~Fleming, B.~Kouvaritakis, and M.~Cannon, ``Robust tube {MPC} for linear
  systems with multiplicative uncertainty,'' \emph{IEEE Transactions on
  Automatic Control}, vol.~60, no.~4, pp. 1087--1092, Apr. 2015, doi:
  10.1109/TAC.2014.2336358.

\bibitem{Zhang2020}
K.~Zhang and Y.~Shi, ``Adaptive model predictive control for a class of
  constrained linear systems with parametric uncertainties,''
  \emph{Automatica}, vol. 117, p. 108974, Jul. 2020, doi:
  10.1016/j.automatica.2020.108974.

\bibitem{Tanaskovic2019b}
M.~Tanaskovic, L.~Fagiano, and V.~Gligorovski, ``{Adaptive model predictive
  control for linear time varying MIMO systems},'' \emph{Automatica}, vol. 105,
  pp. 237--245, Jul. 2019, doi: 10.1016/j.automatica.2019.03.030.

\bibitem{Lorenzen2019}
M.~Lorenzen, M.~Cannon, and F.~Allg{\"{o}}wer, ``Robust {MPC} with recursive
  model update,'' \emph{Automatica}, vol. 103, pp. 461--471, May 2019.

\bibitem{Lu2021}
X.~Lu, M.~Cannon, and D.~Koksal-Rivet, ``{Robust adaptive model predictive
  control: Performance and parameter estimation},'' \emph{International Journal
  of Robust and Nonlinear Control}, vol.~31, no.~18, pp. 8703--8724, Aug. 2021,
  doi: 10.1002/rnc.5175.

\bibitem{Bujarbaruah2022}
M.~Bujarbaruah, U.~Rosolia, Y.~R. St{\"{u}}rz, X.~Zhang, and F.~Borrelli,
  ``Robust {MPC} for {LPV} systems via a novel optimization-based constraint
  tightening,'' \emph{Automatica}, vol. 143, p. 110459, Sep. 2022, doi:
  10.1016/j.automatica.2022.110459.

\bibitem{Chisci1998}
L.~Chisci, A.~Garulli, A.~Vicino, and G.~Zappa, ``Block recursive
  parallelotopic bounding in set membership identification,''
  \emph{Automatica}, vol.~34, no.~1, pp. 15--22, Jan. 1998, doi:
  10.1016/S0005-1098(97)00160-X.

\bibitem{Blanchini2015}
F.~Blanchini and S.~Miani, \emph{Set-Theoretic Methods in Control}, ser.
  Systems {\&} Control: Foundations {\&} Applications.\hskip 1em plus 0.5em
  minus 0.4em\relax Springer International Publishing, 2015.

\bibitem{rawlings2017model}
J.~B. Rawlings, D.~Q. Mayne, and M.~Diehl, \emph{Model predictive control:
  theory, computation, and design}.\hskip 1em plus 0.5em minus 0.4em\relax Nob
  Hill Publishing Madison, WI, 2017, vol.~2.

\bibitem{limon2009input}
D.~Limon, T.~Alamo, D.~M. Raimondo, D.~Pena, J.~M. Bravo, A.~Ferramosca, and
  E.~F. Camacho, ``Input-to-state stability: a unifying framework for robust
  model predictive control,'' in \emph{Nonlinear model predictive
  control}.\hskip 1em plus 0.5em minus 0.4em\relax Springer, 2009, pp. 1--26.

\bibitem{ETHZurich2020}
\BIBentryALTinterwordspacing
E.~Zurich, ``Race task description cybathlon 2020 global edition,'' 2020.
  [Online]. Available:
  \url{https://cybathlon.ethz.ch/documents/downloads/CYBATHLON\_global\_edition\_Races\_and\_Rules.pdf}
\BIBentrySTDinterwordspacing

\bibitem{BrainLine2008}
\BIBentryALTinterwordspacing
BrainLine, ``Wheelchair ramp information,'' 2008. [Online]. Available:
  \url{https://www.brainline.org/article/wheelchair-ramp-information}
\BIBentrySTDinterwordspacing

\bibitem{Gulati2008}
S.~Gulati and B.~Kuipers, ``{High performance control for graceful motion of an
  intelligent wheelchair},'' \emph{Proceedings - IEEE International Conference
  on Robotics and Automation}, pp. 3932--3938, 2008, doi:
  10.1109/ROBOT.2008.4543815.

\end{thebibliography}
	
\end{document}